%% file: eventrees.tex
\RequirePackage{ifpdf}
\ifpdf
 	\documentclass[a4paper,11pt,pdftex]{amsart}
 	\else
 	\documentclass[a4paper,11pt,dvips]{amsart}
\fi

\usepackage[all]{xy}
\ifpdf
\else
  \xyoption{ps}
  \xyoption{dvips}
\fi
\usepackage{ae,aecompl}
\usepackage[english]{babel}
\usepackage[utf8]{inputenc}
\usepackage{palatino}
\usepackage{enumerate} 
\usepackage{amsthm}
\usepackage{amsmath}
\usepackage{amsfonts}
\usepackage{subfig}
\usepackage{graphicx}

\graphicspath{{./figures/}{./xypics/}{./figures.eventrees/}}
\ifpdf
	\usepackage{epstopdf}		
	\DeclareGraphicsExtensions{.png,.jpg,.eps,.epsf}
\fi

\input{./sections/definitions.tex}

\newtheorem{lemma}{Lemma}
\newtheorem{theorem}[lemma]{Theorem}
\newtheorem{corollary}[lemma]{Corollary}
\newtheorem{definition}[lemma]{Definition}

\newtheorem{proposition}[lemma]{Proposition}

\title[On eigenvalues of the Schr\"odinger operator]{On eigenvalues of the Schr\"odinger operator with an even complex-valued polynomial potential}

\author[P.~Alexandersson]{Per Alexandersson}
\address{Department of Mathematics, Stockholm University, SE-106 91, Stockholm, Sweden}
\email{per@math.su.se}

\date{\today}

\begin{document}


\begin{abstract}
In this paper, we generalize several results of the article ``Analytic continuation of eigenvalues of a quartic oscillator'' 
of A.~Eremenko and A.~Gabrielov.

We consider a family of eigenvalue problems for a Schr\"odinger equation with even polynomial potentials
of arbitrary degree $d$ with complex coefficients, and $k<(d+2)/2$ boundary conditions.
We show that the spectral determinant in this case consists of two components, containing
even and odd eigenvalues respectively.

In the case with $k=(d+2)/2$ boundary conditions,
we show that the corresponding parameter space consists of infinitely many connected components.
\end{abstract}

\maketitle

\tableofcontents
\input{./sections.eventrees/intro}
\input{./sections.eventrees/preliminaries}
\input{./sections.eventrees/actions}
\input{./sections.eventrees/invariants}
\input{./sections.eventrees/actionexample.tex}

\newpage

\bibliographystyle{alpha}
\bibliography{/home/paxinum/public/resources/latex/bibliography} 

\end{document}


\pagestyle{empty}

\xymatrix @-1pc {
\cdot & \cdots & \cdot & & \cdot & \cdots & \cdot & & \cdot & \cdots & \cdot & & \cdot & \cdots & \cdot\\
 &\bullet \ar@{..}[l]  \ar@{--}[ul] \ar@{--}[u] \ar@{--}[ur] & & & 
 &\bullet \ar@{--}[ul] \ar@{--}[u] \ar@{--}[ur] \ar@{->}[llll]^{\edge_{j}} & & &
 &\bullet_\junc \ar@{--}[ul] \ar@{--}[u] \ar@{--}[ur] \ar@{->}[llll]^{\edge_{j}} & & &
 &\bullet \ar@{..}[r] \ar@{--}[ul] \ar@{--}[u] \ar@{--}[ur] \ar@{->}[llll]^{\edge_{j_+}} & & & \\
\\
& & & & & w_{j} & & &  & \bullet \ar@/^/@{->}[uu]^{\edge_{j}} \ar@/_/@{<-}[uu]_{\edge_{j_+}} & & w_{j_+} & & \\
\\
& & & &  &  & & &  & \bullet \ar@/^/@{->}[uu]^{\edge_{j}} \ar@/_/@{<-}[uu]_{\edge_{j_+}} & &   \\
& & & & & & & &  &  \ar@{..}[u] & & & & \\
}


\pagestyle{empty}

\xymatrix @-1pc {
\cdot & \cdots & \cdot & & \cdot & \cdots & \cdot & & \cdot & \cdots & \cdot & & \cdot & \cdots & \cdot\\
 &\bullet \ar@{..}[l]  \ar@{--}[ul] \ar@{--}[u] \ar@{--}[ur] & & & 
 &\bullet \ar@{--}[ul] \ar@{--}[u] \ar@{--}[ur] \ar@{->}[llll]^{\edge_{j}} & & &
 &\bullet_\junc \ar@{--}[ul] \ar@{--}[u] \ar@{--}[ur] \ar@{->}[llll]^{\edge_{j_+}} & & &
 &\bullet \ar@{..}[r] \ar@{--}[ul] \ar@{--}[u] \ar@{--}[ur] \ar@{->}[llll]^{\edge_{j_+}} & & & \\
\\
& & & &  & \bullet \ar@/^/@{->}[uu]^{\edge_{j}} \ar@/_/@{<-}[uu]_{\edge_{j_+}} & & & & \\
\\
& & w_{j} & &  & \bullet \ar@/^/@{->}[uu]^{\edge_{j}} \ar@/_/@{<-}[uu]_{\edge_{j_+}} & & &&& w_{j_+} \\
& & & &  &  \ar@{..}[u] & & & & \\
}


\pagestyle{empty}

\xymatrix @-1pc {
\cdot & \cdots & \cdot & & \cdot & \cdots & \cdot & & \cdot & \cdots & \cdot & & \cdot & \cdots & \cdot\\
 &\bullet \ar@{..}[l]  \ar@{--}[ul] \ar@{--}[u] \ar@{--}[ur] & & & 
 &\bullet \ar@{--}[ul] \ar@{--}[u] \ar@{--}[ur] \ar@{->}[llll]^{\edge_{j_+}} & & &
 &\bullet_\junc \ar@{--}[ul] \ar@{--}[u] \ar@{--}[ur]  & & &
 &\bullet \ar@{..}[r] \ar@{--}[ul] \ar@{--}[u] \ar@{--}[ur] \ar@{->}[llll]^{\edge_{j}} & & & \\
\\
& & & w_{j_+} & & &  & \bullet \ar@{->}[uull]^{\edge_{j_+}} \ar@{<-}[uurr]_{\edge_{j}} & & & w_{j}  & \\
\\
& &  &  & & &  & \bullet \ar@/^/@{->}[uu]^{\edge_{j_+}} \ar@/_/@{<-}[uu]_{\edge_{j}} & &   \\
& & & & & &  &  \ar@{..}[u] & & & & \\
}


\pagestyle{empty}

\xymatrix @-1pc {
\cdot & \cdots & \cdot & & \cdot & \cdots & \cdot & & \cdot & \cdots & \cdot & & \cdot & \cdots & \cdot\\
 &\bullet \ar@{..}[l]  \ar@{--}[ul] \ar@{--}[u] \ar@{--}[ur] & & & 
 &\bullet \ar@{--}[ul] \ar@{--}[u] \ar@{--}[ur] \ar@{->}[llll]^{\edge_{j}} & & &
 &\bullet_\junc \ar@{--}[ul] \ar@{--}[u] \ar@{--}[ur] \ar@{->}[llll]^{\edge_{j}} & & &
 &\bullet \ar@{..}[r] \ar@{--}[ul] \ar@{--}[u] \ar@{--}[ur] \ar@{->}[llll]^{\edge_{j_+}} & & & \\
\\
&&&& & & & &  & \bullet \ar@/^/@{->}[uu]^{\edge_{j}} \ar@/_/@{<-}[uu]_{\edge_{j_+}} & & & & \\
&&&&&w_j & & & & \bullet \ar@{..}[u] ^k &&&&w_{j_+}\\
\\
&&&& &  & & &  \bullet \ar@{->}[uur]^{\edge_{j}} &  & \bullet \ar@{<-}[uul]_{\edge_{j_+}} & & & & & & \\
&&&&& & & \ar@{..}[ur] &  & w_{j+1}=0  & & \ar@{..}[ul] & & \\
}


\pagestyle{empty}

\xymatrix @-1pc {
\cdot & \cdots & \cdot & & \cdot & \cdots & \cdot & & \cdot & \cdots & \cdot & & \cdot & \cdots & \cdot\\
 &\bullet \ar@{..}[l]  \ar@{--}[ul] \ar@{--}[u] \ar@{--}[ur] & & & 
 &\bullet \ar@{--}[ul] \ar@{--}[u] \ar@{--}[ur] \ar@{->}[llll]^{\edge_{j_+}} & & &
 &\bullet_\junc \ar@{--}[ul] \ar@{--}[u] \ar@{--}[ur] & & &
 &\bullet \ar@{..}[r] \ar@{--}[ul] \ar@{--}[u] \ar@{--}[ur] \ar@{->}[llll]^{\edge_{j}} & & & \\
&& &     & & & & \bullet  \ar@{->}[ull]^{\edge_{j_+}}  \ar@{<-}[urr]_{\edge_{j}} &&&&&&&&\\
\\
&& & & & &  & \bullet \ar@/^/@{->}[uu]^{\edge_{j_+}} \ar@/_/@{<-}[uu]_{\edge_{j}} & & & & \\
&&&w_{j_+} & & & & \bullet \ar@{..}[u] ^{k} &&&&w_{j}\\
\\
&& &  & & &  \bullet \ar@{->}[uur]^{\edge_{j_+}} &  & \bullet \ar@{<-}[uul]_{\edge_{j}} & &  & & \\
&& & & & \ar@{..}[ur] &  & w_{j+1}=0  & & \ar@{..}[ul] & & \\
}

%% file: sections/definitions.tex
\newcommand{\C}{\mathbf{C}}
\newcommand{\Cb}{\bar{\C}}
\newcommand{\R}{\mathbf{R}}

\newcommand{\sbraid}{\mathcal{H}}	

\newcommand{\junc}{u}		
\newcommand{\vtex}{v}		

\newcommand{\actA}{A}		
\newcommand{\actB}{B}		
\newcommand{\actE}{E}		

\newcommand{\defin}[1]{\emph{#1}}

\mathchardef\mhyp="2D

%% file: sections.eventrees/intro.tex
\section{Introduction}
We study the problem of analytic continuation of eigenvalues of the Schr\"odinger 
operator with an even complex-valued polynomial potential,
that is, analytic continuation of $\lambda=\lambda(\alpha)$ in the differential equation 
\begin{align}\label{eq:evenschroedinger}
-y'' + P_\alpha(z)y = \lambda y, 
\end{align}
where $\alpha=(\alpha_2,\alpha_4,\dots,\alpha_{d-2})$ and $P_\alpha(z)$ is the even polynomial
\begin{align*}
P_\alpha(z) = z^d + \alpha_{d-2}z^{d-2}+\dots+\alpha_2 z^2.
\end{align*}
The boundary conditions for \eqref{eq:evenschroedinger} are as follows:
Set $n = d+2$ and divide the plane into $n$ disjoint open sectors
$$S_j = \{z \in \C \setminus \{0\} : |\arg z - 2\pi j / n|<\pi/n \}, \quad j=0,1,2,\dots,n-1.$$
The index $j$ should be considered mod $n.$
These are the \defin{Stokes sectors} of the equation (\ref{eq:evenschroedinger}).
A solution $y$ of (\ref{eq:evenschroedinger}) satisfies $y(z)\rightarrow 0$ or $y(z)\rightarrow \infty$ 
as $z \rightarrow \infty$ along each ray from the origin in $S_j,$ see \cite{sibuya}. 
The solution $y$ is called \defin{subdominant} in the first case, 
and \defin{dominant} in the second case.

The main result of this paper is as follows:
\begin{theorem}\label{thm:mainone}
Let $\nu = d/2+1$ and let $J = \{j_1,j_2,\dots,j_{2m}\}$ with $j_{k+m} = j_{k}+\nu$ and $|j_p - j_q|>1$ for $p\neq q.$
Let $\Sigma$ be the set of all $(\alpha,\lambda) \in \C^\nu$ for which 
the equation $-y''+(P_\alpha-\lambda)y=0$ has a solution with
with the boundrary conditions
\begin{align}\label{eq:bddconditions}
y \text{ is } \text{ subdominant in } S_j \mbox{ for all } j\in J
\end{align}
where $P_{\mathbf{\alpha}}(z)$ is an even polynomial of degree $d.$
For $m<\nu/2,$ $\Sigma$ consists of two irreducible connected components.
For $m=\nu/2,$ which can only happen when $d \equiv 2 \mod 4,$ 
$\Sigma$ consists of infinitely many connected components, 
distinguished by the number of zeros of the corresponding solution to \eqref{eq:evenschroedinger}.
\end{theorem}

\subsection{Previous results}
The first study of analytic continuation of $\lambda$ in the complex $\beta$-plane 
for the problem
$$-y''+(\beta z^4+z^2)y=\lambda y, \quad y(-\infty)=y(\infty)=0$$
was done by Bender and Wu \cite{benderwu},
They discovered the connectivity of the sets of odd and even eigenvalues,
rigorous results was later proved in \cite{simon}.
 
In \cite{gabrielov}, the even quartic potential $P_a(z) = z^4 + a z^2$ and the boundary value problem
$$-y'' + (z^4 + a z^2)y = \lambda_a y, \quad y(\infty)=y(-\infty)=0$$
was considered.

The problem has discrete real spectrum for real $a,$ with $\lambda_1 < \lambda_2 < \dots \rightarrow +\infty.$
There are two families of eigenvalues, those with even index and those with odd. If $\lambda_j$ and $\lambda_k$ are two eigenvalues in the same family, then $\lambda_k$ can be obtained from $\lambda_j$ by analytic continuation in the complex $\alpha\mhyp$plane.
Similar results have been found for other potentials, 
such as the PT-symmetric cubic, where $P_\alpha(z) = (iz^3+i\alpha z),$ 
with $y(z)\rightarrow 0,$ as $z\rightarrow \pm \infty$ on the real line. See for example \cite{irreduc}.

\subsection{Acknowledgements}
The author would like to thank Andrei Gabrielov for the introduction to this area of research,
and for enlightening suggestions and improvements to the text. 
Great thanks to Boris Shapiro, my advisor. 

%% file: sections.eventrees/preliminaries.tex
\section{Preliminaries on general theory of solutions to the Schroedinger equation}
We will review some properties for the Schr\"odinger equation with a general polynomial potential. 
In particular, these properties hold for an even polynomial potential.
These properties may also be found in \cite{gabrielov,pergabrielov}.

The general Schroedinger equation is given by
\begin{align}\label{eq:generalschroedinger}
-y'' + P_\alpha(z)y = \lambda y, 
\end{align}
where $\alpha=(\alpha_1,\alpha_2,\dots,\alpha_{d-1})$ and $P_\alpha(z)$ is the polynomial
\begin{align*}
P_\alpha(z) = z^d + \alpha_{d-1}z^{d-1}+\dots+\alpha_1 z.
\end{align*}
We have the associated Stokes sectors
$$S_j = \{z \in \C \setminus \{0\} : |\arg z - 2\pi j / n|<\pi/n \}, \quad j=0,1,2,\dots,n-1,$$
where $n=d+2,$ and index considered mod $n.$
The boundary conditions to \eqref{eq:generalschroedinger} are of the form
\begin{align}\label{eq:generalbddconditions}
y \mbox{ is subdominant in } S_{j_1},S_{j_2},\dots,S_{j_k}
\end{align}
with $|j_p-j_q|>1$ for all $p\neq q.$

Notice that any solution $y \neq 0$ of \eqref{eq:generalschroedinger} is an entire function,
and the ratio $f = y/y_1$ of any two linearly independent solutions  of \eqref{eq:generalschroedinger}
is a meromorphic function with the following properties, (see \cite{sibuya}).

\begin{enumerate}[(I)]
\item For any $j,$ there is a solution $y$ of (\ref{eq:generalschroedinger}) subdominant in the Stokes sector $S_j,$
where $y$ is unique up to multiplication by a non-zero constant.

\item For any Stokes sector $S_j$, we have $f(z) \rightarrow w\in\Cb$
as $z\to\infty$ along any ray in $S_j$. This value $w$ is called
\emph{the asymptotic value} of $f$ in $S_j$.

\item \label{list:different} 
For any $j$, the asymptotic values of $f$ in $S_j$ and
$S_{j+1}$ (index still taken modulo $n$) are distinct. 
Furthermore, $f$ has at least 3 distinct asymptotic values.
 
\item The asymptotic value of $f$ in $S_j$ is zero if and only if $y$ is subdominant in $S_j.$ 
We call such sector \defin{subdominant} for $f$ as well.
Note that the boundary conditions given in \eqref{eq:generalbddconditions} imply that sectors $S_{j_1},\dots,S_{j_k}$ 
are subdominant for $f$ when $y$ is an eigenfunction of \eqref{eq:generalschroedinger}, \eqref{eq:generalbddconditions}.

\item \label{list:unramified} $f$ does not have critical points, hence
$f:\mathbf{C}\rightarrow \Cb$ is unramified outside the
asymptotic values.
 
\item The Schwartzian derivative $S_f$ of $f$ given by
$$S_f = \frac{f'''}{f'}-\frac{3}{2}\left( \frac{f''}{f'} \right)^2$$
equals $-2(P_\alpha - \lambda).$ Therefore one can recover $P_\alpha$ and $\lambda$ from $f$.
\end{enumerate}
From now on, $f$ denotes the ratio of two linearly independent solutions of \eqref{eq:generalschroedinger}, \eqref{eq:generalbddconditions}.

\subsection{Cell decompositions}
As above, set $n = \deg P + 2$ where $P$ is our polynomial potential and assume that all non-zero asymptotic values of $f$ are distinct and finite.
Let $w_j$ be the asymptotic values of $f$ with an arbitrary 
ordering satisfying the only restriction that if $S_j$ is subdominant, then $w_j=0.$
One can denote by $w_j$ the asymptotic value in the Stokes sector $S_j,$
which will be called the \defin{standard order}, see section \ref{subsec:standardorder}.

Consider the cell decomposition $\Psi_0$ of $\Cb_w$ shown in Fig.~\ref{fig:curves}a.
It consists of closed directed loops $\gamma_j$ starting and ending at $\infty,$ where the index is considered mod $n,$ and $\gamma_j$ is defined only if $w_j\neq 0.$ The loops $\gamma_j$ only intersect at $\infty$ and have no self-intersection other than $\infty.$
Each loop $\gamma_j$ contains a single non-zero asymptotic value $w_j$ of $f.$
For example, for even $n,$ the boundary condition $y\rightarrow 0$ as $z\rightarrow \pm\infty$ for $z \in \R$ implies that $w_0=w_{n/2}=0,$
 so there are no loops $\gamma_0$ and $\gamma_{n/2}.$ 
We have a natural cyclic order of the asymptotic values, namely the order in which a small circle around $\infty$ traversed 
counterclockwise intersects the associated loops $\gamma_j,$ see Fig.~\ref{fig:curves}a.

We use the same index for the asymptotic values and the loops, so define
$$j_+ = j+k \text{ where } k\in\{1,2\} \text{ is the smallest integer such that } w_{j+k}\neq 0.$$
Thus, $\gamma_{j_+}$ is the loop around the next to $w_j$ (in the cyclic order mod $n$) non-zero asymptotic value. Similarly, $\gamma_{j_-}$ is the loop around the previous non-zero asymptotic value.
 
\begin{figure}
  \centering
    \subfloat[(a) $\Psi_0$]{\includegraphics[width=0.48\textwidth]{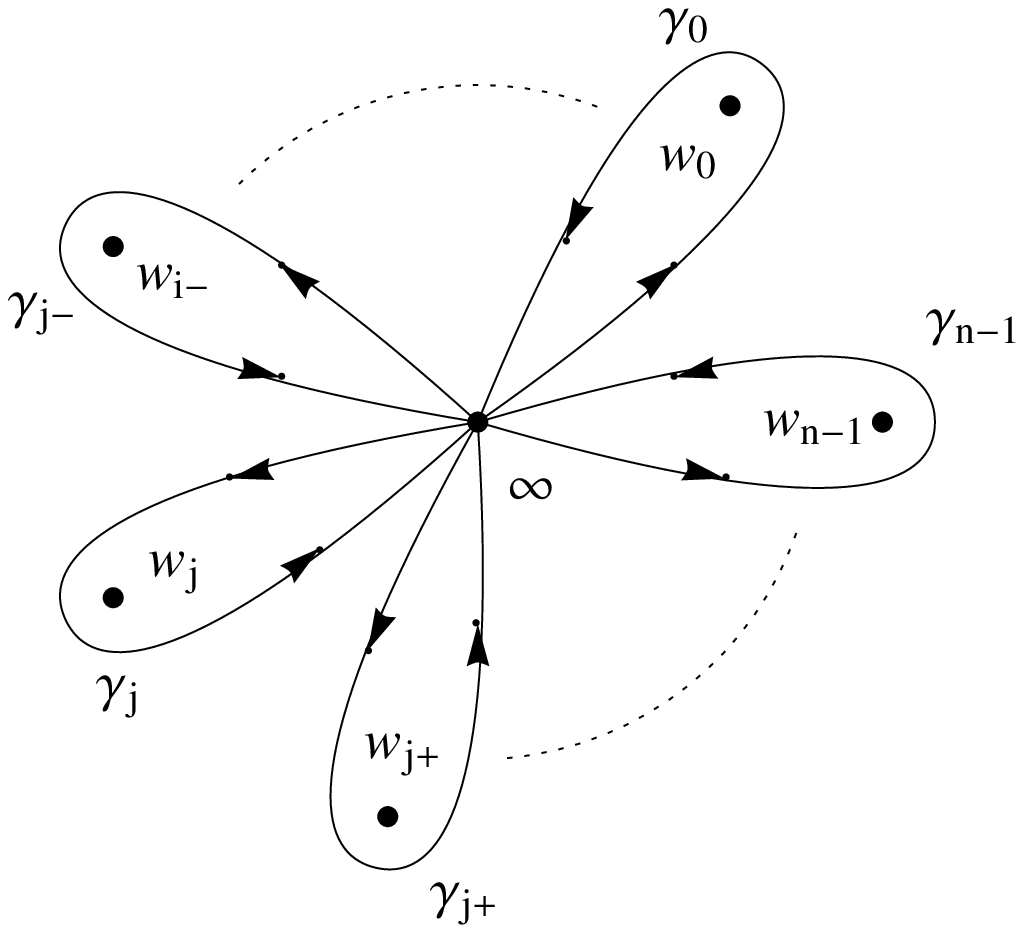}}
    \subfloat[(b) $\actA_j(\Psi_0).$]{\includegraphics[width=0.48\textwidth]{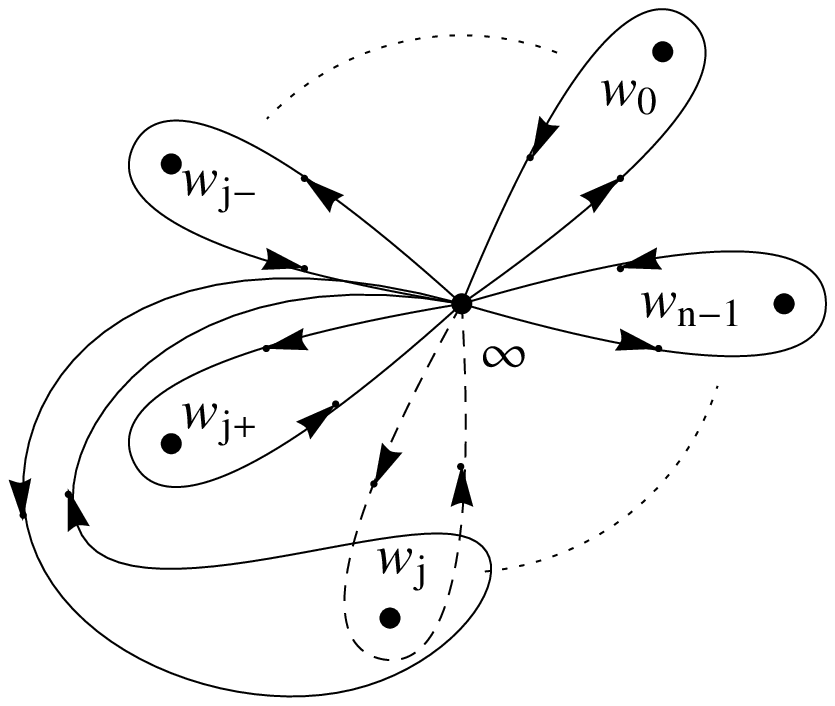}}
    \caption{Permuting $w_j$ and $w_{j_+}$ in $\Psi_0.$}\label{fig:curves}
\end{figure}

\subsection{From cell decompositions to graphs}\label{sec:graphprop}
Proofs of all statements in this subsection can be found in \cite{gabrielov}.
  
Given $f$ and $\Psi_0$ as above, consider the preimage $\Phi_0 = f^{-1}(\Psi_0).$ 
Then $\Phi_0$ gives a cell decomposition of the plane $\C_z.$ 
Its vertices are the poles of $f$ and the edges are preimages of the loops $\gamma_j.$ 
An edge that is a preimage of $\gamma_j$ is labeled by $j$ and called a $j\mhyp$edge.
The edges are directed, their orientation is induced from the orientation of the loops $\gamma_j$.  
Removing all loops of $\Phi_0,$ we obtain an infinite, directed planar graph $\Gamma,$ without loops. 
Vertices of $\Gamma$ are poles of $f,$ each bounded connected component of $\C\setminus \Gamma$ contains one simple zero of $f,$ and each zero of $f$ belongs to one such bounded connected component.
There are at most two edges of $\Gamma$ connecting any two of its vertices. 
Replacing each such pair of edges with a single undirected edge and making all other edges undirected, we obtain an undirected graph $T_\Gamma.$ It has no loops or multiple edges, and the transformation from $\Phi_0$ to $T_\Gamma$ can be uniquely reversed.

A \defin{junction} is a vertex of $\Gamma$ (and of $T_\Gamma$) at which the degree of $T_\Gamma$ is at least 3.
From now on, $\Gamma$ refers to both the directed graph without loops and the associated cell decomposition $\Phi_0$.

\subsection{The standard order of asymptotic values}\label{subsec:standardorder}
For a potential $P$ of degree $d,$ 
the graph $\Gamma$ has $n = d+2$ infinite branches and $n$ unbounded faces corresponding to the Stokes sectors of $P$.
We fixed earlier the ordering $w_0,w_1,\dots,w_{n-1}$ of the asymptotic values of $f.$
 
\emph{If} each $w_j$ is the asymptotic value in the sector $S_{j},$ we say that the asymptotic 
values have \defin{the standard order} and the corresponding cell decomposition $\Gamma$ is a \defin{standard graph}.

\begin{lemma}[See Prop. 6 \cite{gabrielov}]
If a cell decomposition $\Gamma$ is a standard graph, then the corresponding undirected graph $T_\Gamma$ is a tree.
\end{lemma}

In the next section, we define some actions on $\Psi_0$ that permute non-zero asymptotic values.
Each unbounded face of $\Gamma$ (and $T_\Gamma$) will be labeled by the asymptotic value in the corresponding Stokes sector. 
For example, labeling an unbounded face corresponding to $S_k$ with $w_j$ or just with the index $j,$ 
indicates that $w_j$ is the asymptotic value in $S_k.$

From the definition of the loops $\gamma_j,$ a face corresponding to a dominant sector has the same label as any edge bounding that face.
The label in a face corresponding to a subdominant sector $S_k$ is always $k,$ since the actions defined below only permute non-zero asymptotic values.

An unbounded face of $\Gamma$ is called (sub)dominant if the corresponding Stokes sector is (sub)dominant.



\subsection{Properties of graphs and their face labeling}\label{subsec:labelingprop}\label{subsec:junctionoftype}
\begin{lemma}[See Section 3 in \cite{gabrielov}]\label{lemma:labelingprop}
Any graph $\Gamma$ have the following properties:

\begin{enumerate}[(I)]
\item Two bounded faces of $\Gamma$ cannot have a common edge, (since a $j\mhyp$edge is always at the boundary of an unbounded face labeled $j.$)

\item \label{list:numbering} The edges of a bounded face of a graph $\Gamma$ are directed clockwise, 
and their labels increase in that order. 
Therefore, a bounded face of $T_\Gamma$ can only appear if the order of $w_j$ is non-standard. 

\item \label{list:uniquenr} Each label appears at most once in the boundary of any bounded face of $\Gamma.$

\item The unbounded faces of $\Gamma$ adjacent to a junction $\junc,$ 
always have the labels cyclically increasing counterclockwise around $\junc.$

\item \label{list:labelexistone} The boundary of a dominant face labeled $j$ 
consists of infinitely many directed $j\mhyp$edges, oriented counterclockwise around the face.

\item \label{list:labelexisttwo} If $w_j = 0$ there are no $j\mhyp$edges.

\item Each vertex of $\Gamma$ has even degree, since each vertex in $\Phi_0 = f^{-1}(\Psi_0)$ has even degree, 
and removing loops to obtain $\Gamma$ preserves this property.
\end{enumerate}
\end{lemma}

Following the direction of the $j\mhyp$edges, the first vertex that is connected to an edge labeled ${j_+}$ is the vertex where the $j\mhyp$edges and the ${j_+}\mhyp$edges \defin{meet}. The last such vertex is where they \defin{separate}. These vertices, if they exist, must be junctions.
\begin{definition}
Let $\Gamma$ be a standard graph, and let $j \in \Gamma$ be a junction where the $j\mhyp$edges and $j_+\mhyp$edges separate.
Such junction is called a \defin{$j\mhyp$junction.}
\end{definition}
There can be at most one $j\mhyp$junction in $\Gamma,$ the existence of two or more such junctions would violate property (\ref{list:uniquenr}) of the face labeling. 
However, the same junction can be a $j\mhyp$junction for different values of $j.$

There are three different types of $j\mhyp$junctions, see Fig.~\ref{fig:junctiontype_case123}. 
\begin{figure}[ht!]
\centering
  \subfloat[(a)]{\includegraphics{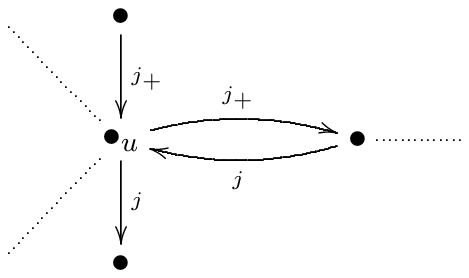}}
  \hfill
  \subfloat[(b)]{\includegraphics{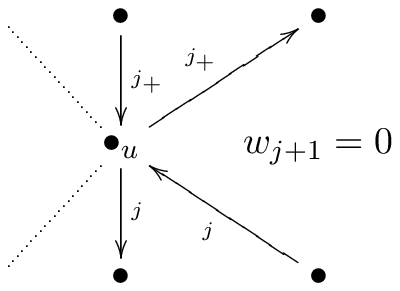}}
  \hfill
  \subfloat[(c)]{\includegraphics{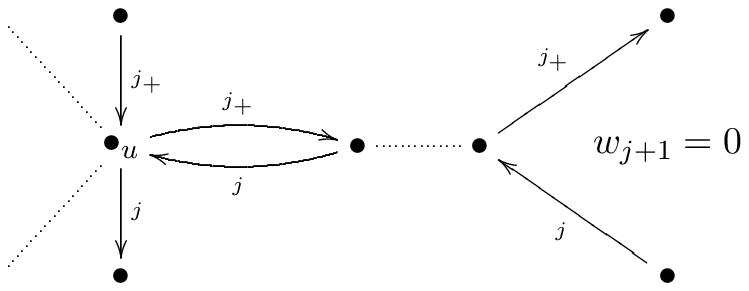}}
  \caption{Different types of $j\mhyp$junctions.}
  \label{fig:junctiontype_case123}
\end{figure}

Case (a) only appears when $w_{j+1}\neq 0.$ 
Cases (b) and (c) can only appear when $w_{j+1}=0.$ 
In (c), the $j\mhyp$edges and $j_+\mhyp$edges meet and separate at different junctions,
while in (b), this happens at the same junction.

\begin{definition}
Let $\Gamma$ be a standard graph with a $j\mhyp$junction $u$.
A \emph{structure} at the $j\mhyp$junction is the subgraph $\Xi$
of $\Gamma$ consisting of the following elements:
\begin{itemize}
\item The edges labeled $j$ that appear before $\junc$ following the $j\mhyp$edges.
\item The edges labeled $j_+$ that appear after $\junc$ following the $j_+\mhyp$edges.
\item All vertices the above edges are connected to.
\end{itemize}
If $\junc$ is as in Fig.~\ref{fig:junctiontype_case123}a, $\Xi$ is called an \defin{$I\mhyp$structure at the $j\mhyp$junction}.
If $\junc$ is as in Fig.~\ref{fig:junctiontype_case123}b,
$\Xi$ is called a \defin{$V\mhyp$structure at the $j\mhyp$junction}.
If $\junc$ is as in Fig.~\ref{fig:junctiontype_case123}c,
$\Xi$ is called a \defin{$Y\mhyp$structure at the $j\mhyp$junction}.
\end{definition}
Since there can be at most one $j\mhyp$junction, there can be at most one structure at the $j\mhyp$junction.

A graph $\Gamma$ shown in Fig.~\ref{fig:example_structs} has one (dotted) $I\mhyp$structure at the $1\mhyp$junction $v,$
one (dotted) $I\mhyp$structure at the $4\mhyp$junction $u,$
one (dashed) $V\mhyp$structure at the $2\mhyp$junction $v$ and one (dotdashed) $Y\mhyp$structure at the $5\mhyp$junction $u$.
\begin{figure}
\centering
\includegraphics[width=0.8\textwidth]{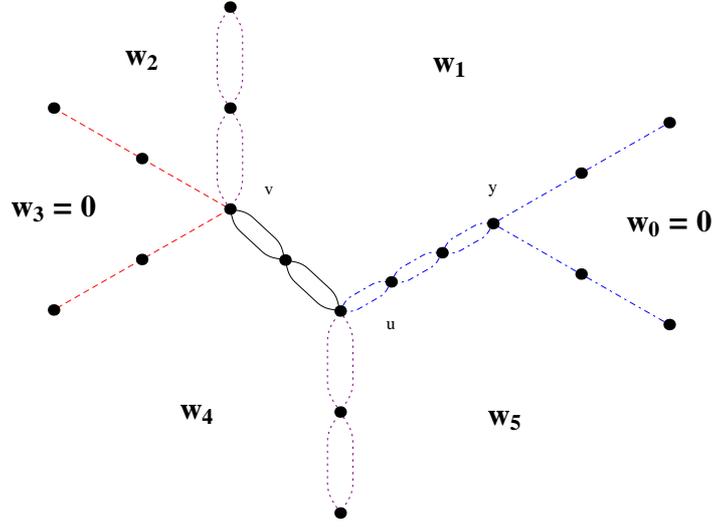}
\caption{Graph $\Gamma$ with (dotted) $I\mhyp$structures, a (dashed) $Y\mhyp$structure and a (dotdashed) $Y\mhyp$structure.}\label{fig:example_structs}
\end{figure}

Note that the $Y\mhyp$structure is the only kind of structure that contains an additional junction.
We refer to such additional junctions as \defin{$Y\mhyp$junctions}. 
For example, the junction marked $y$ in Fig.~\ref{fig:example_structs} is a $Y\mhyp$junction.

\subsection{Braid actions on graphs}\label{subsec:actiondef}

As in \cite{pergabrielov}, we define continuous deformations $\actA_j$ of the loops in Fig.~\ref{fig:curves}a,
such that the new loops are given in terms of the old ones by
$$
\actA_j(\gamma_k) =
\begin{cases}
\gamma_j\gamma_{j_+}\gamma_j^{-1} \text{ if } k=j\\
\gamma_j\text{ if } k=j_+ \\
\gamma_k \text{ otherwise }
\end{cases},
\quad
\actA_j^{-1}(\gamma_k) =
\begin{cases}
\gamma_{j_+} \text{ if } k=j \\
\gamma_{j_+}^{-1}\gamma_{j}\gamma_{j_+} \text{ if } k=j_+\\
\gamma_k \text{ otherwise }
\end{cases}
$$
These actions, together with their inverses, 
generate the Hurwitz (or sphere) braid group $\sbraid_m,$ 
where $m$ is the number of non-zero asymptotic values.
(For a definition of this group, see \cite{lando}.)
The action of the generators $\actA_j$ and $\actA_k$ commute if $|j-k|\geq 2.$

The property (\ref{list:unramified}) of the eigenfunctions implies 
that each $\actA_j$ induces a monodromy transformation of the cell decomposition $\Phi_0,$ and of the associated directed graph $\Gamma.$

%% file: sections.eventrees/actions.tex
\section{Properties of even actions on centrally symmetric graphs}\label{sec:action}

\subsection{Additional properties for even potential}

In addition to the previous properties for general polynomials,
these additional properties holds for even polynomial potentials $P$ (see \cite{gabrielov}).
From now until the end of the article, $\nu = (\deg(P) + 2)/2.$

Each solution $y$ of \eqref{eq:evenschroedinger}
is either even or odd and we may choose $y$ and $y_1$ such that $f=y/y_1$ is odd.

If the asymptotic values $w_0,w_1,\dots,w_{2\nu-1}$ are ordered in the standard order, 
we have that $w_j = -w_{j+\nu}.$

We may choose the loops centrally symmetric in Fig.~\ref{fig:curves}a 
which implies that $\Phi_0$ and $\Gamma$ are centrally symmetric.

\subsection{Even braid actions}
Define the \emph{even actions} $\actE_j$ as $\actE_j = \actA_j \circ \actA_{j+\nu}.$ 

Assume that $\Gamma$ is a graph with the property that if $w_j$ is the asymptotic value in $S_k,$ 
then $w_{j+n}$ is the asymptotic value in $S_{k+\nu}.$ 
(For example, all standard graphs have this property, with $j=k.$)
It follows from the symmetric property of $\actE_j$ that $\actE_j$
preserves this property. To illustrate, 
we have that $\actE_j(\Psi_0)$ is given in Fig.~\ref{fig:curvessymmetric}.
\begin{figure}
\includegraphics[width=0.90\textwidth]{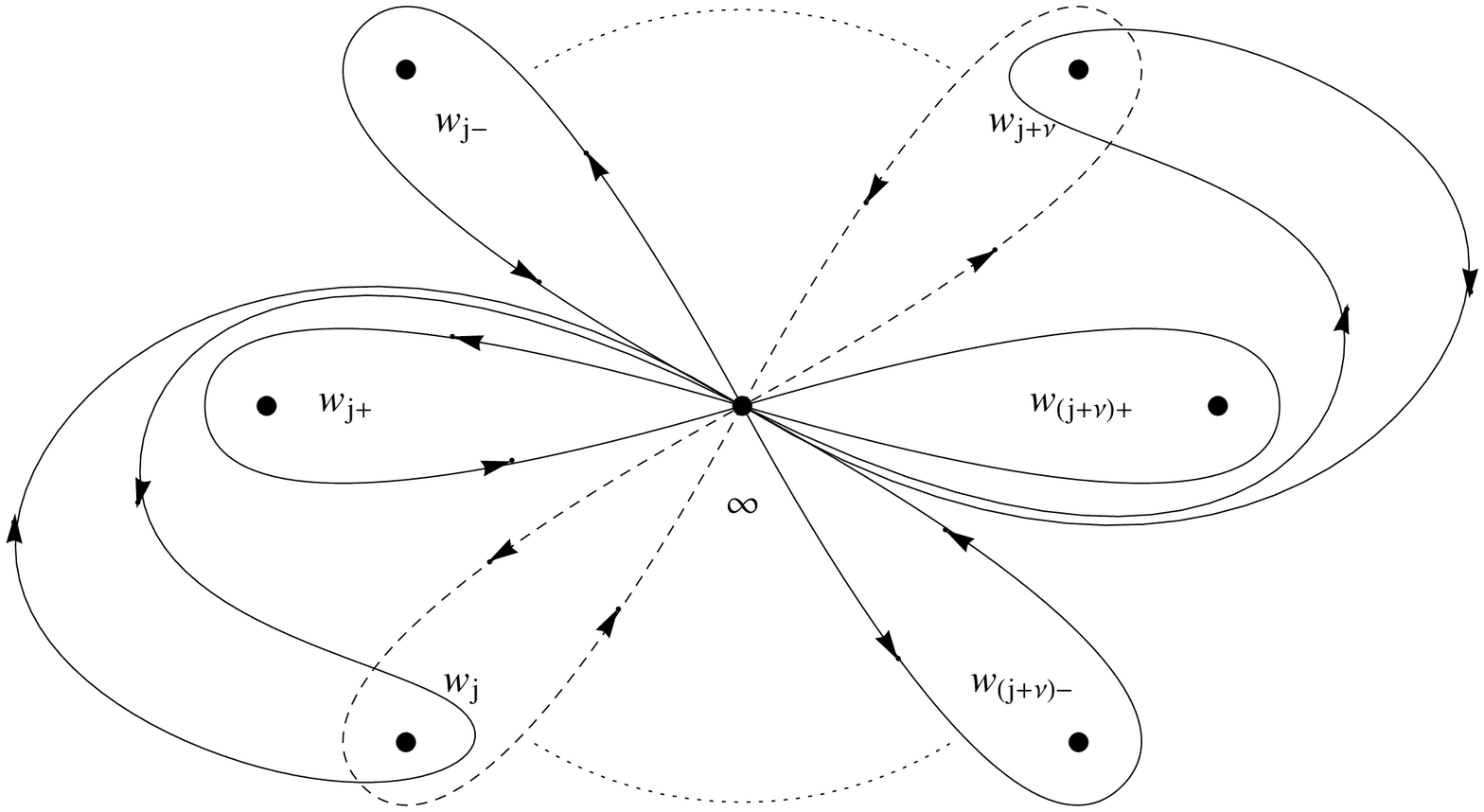}
\caption{$\actE_j(\Psi_0)$}\label{fig:curvessymmetric}
\end{figure}

\begin{lemma}
If $\Gamma$ is centrally symmetric, then $\actE_j(\Gamma)$ and $\actE_ j^{-1}(\Gamma)$ are centrally symmetric graphs.
\end{lemma}
\begin{proof}
We may choose the deformations of the paths $\gamma_j$ and $\gamma_{j+\nu}$ being centrally symmetric,
which implies that the composition $\actA_j \circ \actA_{j+\nu}$ preserves 
the property of $\Gamma$ being centrally symmetric,
see details in \cite{gabrielov}.
\end{proof}

\begin{lemma}\label{lemma:noaction}
Let $\Gamma$ be a centrally symmetric standard graph with no $j\mhyp$junction. Then $\actE_j^2(\Gamma)=\Gamma.$
\end{lemma}
\begin{proof}
Since $\actA_j$ and $\actA_{j+n}$ commute, we have that $\actE_j^2 = \actA_j^2 \actA_{j+\nu}^2,$ 
and the statement then follows from \cite[Lemma 12]{pergabrielov}.
\end{proof}

\begin{theorem}\label{thm:action}
Let $\Gamma$ be a centrally symmetric standard graph with a $j\mhyp$junction $\junc.$
Then $\actE_j^2(\Gamma)\neq \Gamma,$ 
and the structure at the $j\mhyp$junction
is moved one step in the direction of the $j\mhyp$edges under $\actE_j^2.$
The inverse of $\actE^2_j$ moves the structure at the $j\mhyp$junction one step backwards along the $j_+\mhyp$edges.

Since $\Gamma$ is centrally symmetric, it also has a $j+\nu\mhyp$junction,
and the structure at the $j+\nu\mhyp$junction is
moved one step in the direction of the $j+n\mhyp$edges under $\actE_j^2.$
The inverse of $\actE^2_j$ moves the structure at the $j+\nu\mhyp$junction one step backwards along the $(j+\nu)_+\mhyp$edges.
\end{theorem}
\begin{proof}
Since $\actE_j^2 = \actA_j^2 \actA_{j+\nu}^2,$ the result follows from \cite[Theorem 13]{pergabrielov}.
\end{proof}

%% file: sections.eventrees/invariants.tex
\section{Proving Main Theorem \ref{thm:mainone}}

Notice that each centrally symmetric standard graph $\Gamma$ has either a vertex in its center, 
or a double edge, connecting two vertices.
This property follows from the fact that $\Gamma_T$ is a centrally symmetric tree.

\begin{lemma}
Let $\Gamma$ be a centrally symmetric graph. Then for every action $\actE_j,$
$\Gamma$ has a vertex at the center iff $\actE_j(\Gamma)$ has a vertex at the center.
\end{lemma}
\begin{proof}
This is evident from the definition of the actions, since the action only changes the edges,
and preserves the vertices.
\end{proof}
\begin{corollary}
The spectral determinant has at least two connected components.
\end{corollary}

Each centrally symmetric standard graph $\Gamma$ is of one of two types:
\begin{enumerate}
\item $\Gamma$ has a central double edge. The vertices of the central double edge are called \defin{root junctions}.
\item $\Gamma$ has a junction at its center. This junction is called the \defin{root junction} $\junc_r.$
\end{enumerate}

\begin{definition}
A centrally symmetric standard graph $\Gamma$ is in \defin{ivy form} if $\Gamma$ consists of
structures connected to one or two root junctions.
\end{definition}

\begin{definition}
Let $\Gamma$ be a centrally symmetric standard graph.

The \defin{root metric} of $\Gamma,$ denoted $|\Gamma|_{r}$ is defined as
$$|\Gamma|_{r} = \sum_{\vtex \in \Gamma} \left(\deg(\vtex)-2\right)|\vtex - \junc_r|$$
where the sum is taken over all vertices $\vtex$ of $\Gamma_1.$
Here $deg(\vtex)$ is the total degree of the vertex $\vtex$ in $T_\Gamma$ and 
$|\vtex - \junc_r|$ is the length of the shortest path from $\vtex$ 
to the closest root junction $\junc_r$ in $T_\Gamma.$
\end{definition}

\begin{lemma}
The graph $\Gamma$ is in ivy form if and only if all but its root junctions are 
$Y\mhyp$junctions.
\end{lemma}
\begin{proof}
This follows from the definitions of the structures.
\end{proof}

\begin{theorem}\label{thm:toIVY}
Let $\Gamma$ be a centrally symmetric standard graph.
Then there is a sequence of even actions $\actE^* = \actE_{j_1}^{\pm2},\actE_{j_2}^{\pm2},\dots,$ 
such that $\actE^*(\Gamma)$ is in \defin{ivy form}.
\end{theorem}
\begin{proof}
Assume that $\Gamma$ is not in ivy form.

Let $U$ be the set of junctions in $\Gamma$ that are not $Y\mhyp$junctions. 
Since $\Gamma$ is not in ivy form we have that $|U|\geq 3.$
Let $\junc_r \neq \junc_1$ be two junctions in $U$ such that $|\junc_r-\junc_1|$ is maximal,
and $\junc_r$ is the central junction closest to $\junc_1.$
Let $p$ be the path from $\junc_r$ to $\junc_1$ in $T_\Gamma.$
It is unique since $T_\Gamma$ is a tree.
Let $\vtex$ be the vertex preceding $\junc_1$ on the path $p.$ 
The edge from $\vtex$ to $\junc_1$ in $T_\Gamma$ is adjacent to 
at least one dominant face with label $j$ such that $w_j \neq 0.$ 
Therefore, there exists a $j\mhyp$edge between $\vtex$ and 
$\junc_1$ in $\Gamma.$
Suppose first that
this $j\mhyp$edge is directed from $\junc_1$ to $\vtex.$
Let us show that in this case $\junc_1$ must be a $j\mhyp$junction, 
i.e., the dominant face labeled $j_+$ is adjacent to $\junc_1$.

Since $\junc_1$ is not a $Y\mhyp$junction, there is a dominant face
adjacent to $\junc_1$ with a label $k\ne j,j_+$.
Hence no vertices of $p$, except possibly $u_1$ can be
adjacent to  $j_+\mhyp$edges.
If $\junc_1$ is not a $j\mhyp$junction, there are no 
$j_+\mhyp$edges adjacent to $\junc_1$.
This implies that any vertex of $\Gamma$ adjacent to a $j_+\mhyp$edge
is further away from $\junc_r$ than $\junc_1$.

Let $\junc_2$ be the closest to $\junc_1$ vertex of $\Gamma$
adjacent to a $j_+\mhyp$edge.
Then $\junc_2$ should be a junction of $T_\Gamma$, 
since there are two $j_+\mhyp$edges adjacent to $\junc_2$ in $\Gamma$
and at least one more vertex (on the path from $\junc_1$ to $\junc_2$) 
which is connected to $\junc_2$ by edges with labels other than $j_+$.
Since $\junc_2$ is further away from $\junc_r$ that $\junc_1$
and the path $p$ is maximal, $\junc_2$ must be a $Y\mhyp$junction.
If the $j\mhyp$edges and $j_+\mhyp$edges would meet at $\junc_2$,
$\junc_1$ would be a $j\mhyp$junction.
Otherwise, a subdominant face labeled $j+1$
would be adjacent to both $\junc_1$ and $\junc_2$,
while a subdominant face adjacent to a $Y\mhyp$junction
cannot be adjacent to any other junctions.

Hence $\junc_1$ must be a $j\mhyp$junction.
By Theorem \ref{thm:action}, the action $\actE_{j}^{2}$ 
moves the structure at the $j\mhyp$junction $\junc_1$
one step closer to $\junc_r$ along the path $p,$
and similarly happens on the opposite side of $\Gamma,$
decreasing $|\Gamma|_{c}$ by at least 2.

The case when the $j\mhyp$edge is directed from $\vtex$ to $\junc_1$
is treated similarly. In that case,
$\junc_1$ must be a $j_-\mhyp$junction, and the action $\actA_{j_-}^{-2}$ moves 
the structure at the $j_-\mhyp$junction $\junc_1$ 
one step closer to $\junc_r$ along the path $p.$

We have proved that if $|U|>1$ then $|\Gamma|_{r}$ can be reduced. 
Since it is a non-negative integer,
after finitely many steps we must reach a stage where $U$ consists only of the root junctions.
Hence $\actE^*(\Gamma)$ is in ivy form.
\end{proof}

The above Theorem shows that for every centrally symmetric standard graph $\Gamma,$
there is a sequence of actions that turns $\Gamma$ into ivy form. 
A graph in ivy form consists of one or two root junctions, with attached structures.
These structures can be ordered counterclockwise around each root junction.
These observations motivates the following lemmas:

\begin{lemma}\label{lemma:contracttwo}
Let $\Gamma$ be a centrally symmetric standard graph, 
and let $\junc_r \in \Gamma$ be a root junction of type $j_-$ and of type $j.$
Let $S_1$ and $S_2$ be the corresponding structures attached to $\junc_r.$
\begin{enumerate}
\item If $S_1$ and $S_2$ are of type $Y$ resp. $V,$ 
then there is a sequence of even actions that interchange these structures.
\item If $S_1$ and $S_2$ are of type $I$ resp. $Y,$ 
then there is a sequence of even actions that converts the type $Y$ structure to a type $V$ structure.
\item If $S_1$ and $S_2$ are both of type $Y,$
then there is a sequence of even actions that converts one of the $Y\mhyp$structures to a $V\mhyp$structure.
\end{enumerate}
\end{lemma}
\begin{proof}
By symmetry, there are identical structures in $\Gamma$ attached to a root junction of type $n+j_-$ and $n+j,$
with attached structures $S'_1$ and $S'_2$ of the same type as $S_1$ resp. $S_2.$

Lemma 19, 20 and 22 in \cite{pergabrielov}, 
gives the existence of a non-even sequence of actions, 
that only acts on $S_1$ and $S_2$ in the desired way.

In all these cases, the sequence is of the form
$$\actA^* = \actA_{k_1}^{\pm 2} \actA_{k_2}^{\pm 2}\dots \actA_{k_m}^{\pm 2} $$
where $k_1,k_2,\dots k_m \in \{ j,j_- \}.$
It follows that the action 
$$\actB^* = \actA_{k_1+\nu}^{\pm 2} \actA_{k_2+\nu}^{\pm 2}\dots \actA_{k_m+\nu}^{\pm 2} $$
do the same as $\actA^*$ but on $S'_1$ and $S'_2.$

Now, $\actE^* = \actA^* \circ \actB^*$ is even, since by commutativity\footnote{We have at least 4 structures, 2 of them are Y or V structures. Hence $n\geq 6$ and we have commutativity.}, it is equal to
$$(\actA_{k_1}^{\pm 2}\actA_{k_1+\nu}^{\pm 2}) (\actA_{k_2}^{\pm 2}\actA_{k_2+\nu}^{\pm 2}) \dots (\actA_{k_m}^{\pm 2}\actA_{k_m+\nu}^{\pm 2})$$
which easily may be written in terms of our even actions as
$$\actE_{k_1}^{\pm 2} \actE_{k_2}^{\pm 2}\dots \actE_{k_m}^{\pm 2}.$$
This sequence of actions has the desired property.
\end{proof}

\begin{corollary}
Let $\Gamma$ be a centrally symmetric graph, with two adjacent dominant faces.
Then there is a sequence of even actions $\actE^*$ such that $\actE^*(\Gamma)$ has either one or two junctions.
\end{corollary}
\begin{proof}
We may apply even actions to make $\Gamma$ into a standard graph, and then convert it to ivy form.
The condition that we have two dominant faces, is equivalent to existence of $I\mhyp$structures.
If there are no $Y\mhyp$structures, then the only junctions of $\Gamma$ are the root junctions, and we are done.
Otherwise, we may the $Y\mhyp$ and $V\mhyp$structures, so that a $Y\mhyp$structure appears next to the $I\mhyp$structure.
By using the second part of the above lemma, we decrease the number of $Y\mhyp$structures of $\Gamma$ by two.
After a finite number of actions, we arrive at a graph in ivy form without $Y\mhyp$structures.
\end{proof}

\begin{lemma}\label{lemma:toTwoY}
Let $\Gamma$ be a centrally symmetric graph, with no adjacent dominant faces.
Then there is a sequence of even actions $\actE^*$ such that $\actE^*(\Gamma)$ is in ivy form, 
with at most two $Y\mhyp$structures.
\end{lemma}
\begin{proof}
By Teorem \ref{thm:toIVY}, we may assume that $\Gamma$ is in ivy form.
Since there are no adjacent dominant sectors, the only structures of $\Gamma$ are of $Y$ and $V$ type.
These are attached to the one or two root junctions.

Assume that there are more than two $Y\mhyp$structures present.
Two of these must be attached to the same root junction, $\junc_r.$
By repeatedly applying part one of Lemma \ref{lemma:contracttwo}, 
we may interchange the $Y\mhyp$ and $V\mhyp$structures attached to $\junc_r$
such that the two $Y\mhyp$structures are adjacent.
Applying part three of Lemma \ref{lemma:contracttwo}, 
we may then convert one of the two $Y\mhyp$structures to a $V\mhyp$structure.

By symmetry, the same change is done on the opposite side of $\Gamma$
and total number of $Y\mhyp$structures of $\Gamma$ have therefore been reduced by two.
We may repeat this procedure a finite number of times,
until the number of $Y\mhyp$structures is less than three.
This implies the lemma.
\end{proof}

\begin{lemma}[See \cite{pergabrielov}]\label{lemma:boundedfaces}
Let $\Gamma$ be a standard graph such that no two dominant faces are adjacent.
Then the number of bounded faces of $\Gamma$ is finite
and does not change after any action $\actA_j^2$.
\end{lemma}
\begin{corollary}
The number of bounded faces of $\Gamma$ does not change under any \emph{even} action $\actE_j^2 = \actA_j^2 \actA_{j+\nu}^2.$
\end{corollary}

\begin{lemma}\label{lemma:smooth}
Let $\nu = n/2=d/2+1$ and let $\Sigma$ be the space of all $(\alpha,\lambda)\in\C^{\nu-1}$ 
such that equation \eqref{eq:evenschroedinger} admits a solution subdominant
in non-adjacent Stokes sectors 
\begin{align}\label{eq:smoothone}
S_{j_1},S_{j_2},\dots,S_{2m}
\end{align}
with $j_{k+m} = j_{k}+\nu$ and $1\leq m \leq \nu/2.$
Then $\Sigma$ is a smooth complex analytic submanifold of $\C^{\nu-1}$
of the codimension $m.$
\end{lemma}

\begin{proof}
We consider the space $\C^{\nu-1}$ as a subspace of the space $\C^{n-2}$ of
all $(\alpha,\lambda)$ corresponding to the general 
polynomial potentials in \eqref{eq:generalschroedinger}, with $\alpha=(\alpha_1,\dots,\alpha_{d-1})$.
Let $f$ be a ratio of two linearly independent solutions of \eqref{eq:generalschroedinger},
and let $w=(w_0,\dots,w_{n-1})$ be the set of the asymptotic values of $f$
in the Stokes sectors $S_0,\dots,S_{n-1}$.

Then $w$ belongs to the subset $Z$ of ${\Cb}^{n-1}$
where the values $w_j$ in adjacent Stokes sectors are distinct and there are
at least three distinct values among $w_j$.
The group $G$ of fractional-linear transformations of $\Cb$ acts on $Z$ diagonally, 
and the quotient $Z/G$ is a $(n-3)$-dimensional complex manifold.

Theorem 7.2, \cite{bakken} implies that the mapping $W:\C^{n-2}\to Z/G$
assigning to $(\alpha,\lambda)$ the equivalence class of $w$ is submersive.
More precisely, $W$ is locally invertible on the subset $\{\alpha_{d-1}=0\}$ of $\C^{n-2}$

For an even potential, there exists an odd function $f.$
The corresponding set of asymptotic values satisfies $\nu$ linear conditions
$w_{j+\nu}=-w_j$ for $j=0,\dots,\nu-1$.
For $(\alpha,\lambda)\in\Sigma$, we can assume that $S_{j_1},\dots,S_{j_m}$
are subdominant sectors for $f$.
This adds $m$ linearly independent conditions
$w_{j_1}=\dots=w_{j_m}=0.$
Let $Z_0$ be the corresponding subset of $Z$.
Its codimension in $Z$ is $\nu+m$.
The one-dimensional subgroup $\C^*$ of $G$ consisting of multiplications
by non-zero complex numbers preserves $Z_0$,
and $gZ_0\cap Z_0=\emptyset$ for each $g\in G\setminus \C^*$.
The explaination is as follows:

Since we have at least two subdominant sectors,
only fractional linear transforms that preserves 0 are allowed.
Furthermore, there exists a
sector $S_k$ with the value $w_k$ different from 0 and $\infty$
(otherwise we would have only two asymptotic values).
There is a unique transformation, multiplication by $w_k^{-1}$,
preserving 0 and sending $\pm w_k$ to $\pm 1$. This implies that
the only transformation preserving 0 and sending $\pm w_k$
to another pair of opposite numbers is multiplication by
a non-zero constant.

Hence $GZ_0$ is a $G$-invariant submanifold of $Z$ of codimension $\nu+m-2$,
and its image $Y_0\subset Y$ is a smooth submanifold of codimension $\nu+m-2$.
Due to Bakken's theorem, $W^{-1}(Y_0)$ intersected with the $(n-3)$-dimensional
space of $(\alpha,\lambda)$ with $\alpha_{d-1}=0$ is a smooth submanifold
of codimension $\nu+m-2$, dimension $\nu-m-1$.
Accordingly, it is a smooth submanifold of codimension $m$
of the space $\C^{\nu-1}$. 
\end{proof}

\begin{proposition}\label{thm:connected}
Let $\Sigma$ be as in Lemma \ref{lemma:smooth}.
If at least two adjacent Stokes sectors are missing in \eqref{eq:smoothone},
then $\Sigma$ consists of two irreducible complex analytic manifolds.
\end{proposition}
\begin{proof}

Nevanlinna theory (see \cite{nevanlinnaU,nevanlinnaF}),
implies that, for each symmetric standard graph $\Gamma$
with the properties listed in Lemma \ref{lemma:labelingprop},
there exists $(\alpha,\lambda)\in\C^{n-1}$ and an \emph{odd} meromorphic function $f(z)$ such that
$f$ is the ratio of two linearly independent solutions of (\ref{eq:evenschroedinger})
with the asymptotic values $w_j$ in the Stokes sectors $S_j$, and $\Gamma$ is the graph
corresponding to the cell decomposition $\Phi_0=f^{-1}(\Psi_0)$.
This function, and the corresponding point $(\alpha,\lambda)$ is defined uniquely.

Let $W:\Sigma\to Y_0$ be as in the proof of Lemma \ref{lemma:smooth}.
Then $\Sigma$ is an unramified covering of $Y_0$. 
Its fiber over the equivalence class of $w\in Y_0$ consists of the points
$(\alpha_\Gamma,\lambda_\Gamma)$ for all standard graphs $\Gamma$.
Each action $\actA_j^2$ corresponds to a closed loop in 
$Y_0$ starting and ending at $w$. It should be noted that $Y_0$ is a connected manifold.
Since for a given list of subdominant sectors
a standard graph with one vertex is unique,
Theorem \ref{lemma:contracttwo} implies that the monodromy group has two orbits;
odd and even eigenfunctions cannot be exchanged by any path in $Y_0$, 
while any odd (even) can be transferred into any other odd (even) eigenfunction
by a sequence of $E_k^{\pm 2}$.

Hence $\Sigma$ consists of two irreducible connected components (see, e.g., \cite{khovanskii}).
\end{proof}
This immediately implies Theorem \ref{thm:mainone}, for $m < \nu/2.$
The following propostion implies the case where $m = \nu/2.$

\begin{proposition}\label{thm:disconnected}
Let $\Sigma$ be the space of all $(\alpha,\lambda)\in\C^{\nu-1}$, for even $\nu,$
such that equation \eqref{eq:evenschroedinger} admits a solution subdominant in every other Stokes 
sector, that is, in $S_0,S_{2},\dots,S_{n-2}.$

Then irreducible components $\Sigma_k,$ $k=0,\,1,\dots$ of $\Sigma$, which are also its connected
components, are in one-to-one correspondence with the sets of centrally symmetric standard
graphs with $k$ bounded faces.
The corresponding solution of \eqref{eq:evenschroedinger} has $k$ zeros
and can be represented as $Q(z)e^{\phi(z)}$ where $Q$ is a polynomial
of degree $k$ and $\phi$ a polynomial of degree $(d+2)/2$.
\end{proposition}
\begin{proof}
Let us choose $w$ and $\Psi_0$ as in the proof of Proposition \ref{thm:connected}.
Repeating the arguments in the proof of Proposition \ref{thm:connected}, 
we obtain an unramified covering $W:\Sigma\to Y_0$ such that its fiber over $w$ 
consists of the points $(\alpha_\Gamma,\lambda_\Gamma)$ for all standard graphs
$\Gamma$ with the properties listed in Lemma \ref{lemma:labelingprop}.

Since we have no adjacent dominant sectors, Lemma \ref{lemma:toTwoY}
implies that any standard graph $\Gamma$ can be transformed by the monodromy 
action to a graph $\Gamma_0$ in ivy form with at most two $Y$-structures
attached at the root junction(s) of type $j$ and $j+\nu.$

Lemma \ref{lemma:boundedfaces} implies that
$\Gamma$ and $\Gamma_0$ have the same number $k$ of bounded faces.
If $k=0$, the graph $\Gamma_0$ is unique.
If $k>0$, the graph $\Gamma_0$ is completely determined by $k.$
Hence for each $k=0,1,\dots$ there is a unique
orbit of the monodromy group action on the fiber of $W$
over $w$ consisting of all standard graphs $\Gamma$ with $k$ bounded faces.
This implies that $\Sigma$ has one irreducible component
for each $k$. 

Since $\Sigma$ is smooth by Lemma \ref{lemma:smooth}, its irreducible 
components are also its connected components.

Finally, let $f_\Gamma=y/y_1$ where $y$ is an odd solution of \eqref{eq:evenschroedinger}
subdominant in the Stokes sectors $S_0,S_2,\dots,S_{n-2}$.
Then the zeros of $f$ and $y$ are the same, each such zero belongs to
a bounded domain of $\Gamma$, and each bounded domain of $\Gamma$ contains
a single zero. Hence $y$ has exactly $k$ simple zeros.
Let $Q$ be a polynomial of degree $k$ with the same zeros as $y$.
Then $y/Q$ is an entire function of finite order without zeros,
hence $y/Q=e^\phi$ where $\phi$ is a polynomial.
Since $y/Q$ is subdominant in $(d+2)/2$ sectors, $\deg\phi=(d+2)/2$. 
\end{proof}

%% file: sections.eventrees/actionexample.tex
\section{Illustrating example}

We will now give a small example on how to apply the method given in the previous section, 
Theorem \ref{thm:toIVY} and Lemma \ref{lemma:contracttwo}.
Let $\Gamma$ be as in Fig.~\ref{fig:action1}a. From subsection \ref{subsec:labelingprop}, we have that
a dominant face with label $j$ have $j\mhyp$edges as boundaries. 
Hence the faces $0$ and $4$ are subdominant. 
Also, the direction of the edges are directed counterclockwise in each of the dominant faces.
\begin{figure}[!ht]
\centering
  \subfloat[(a)]{\includegraphics[width=0.49\textwidth]{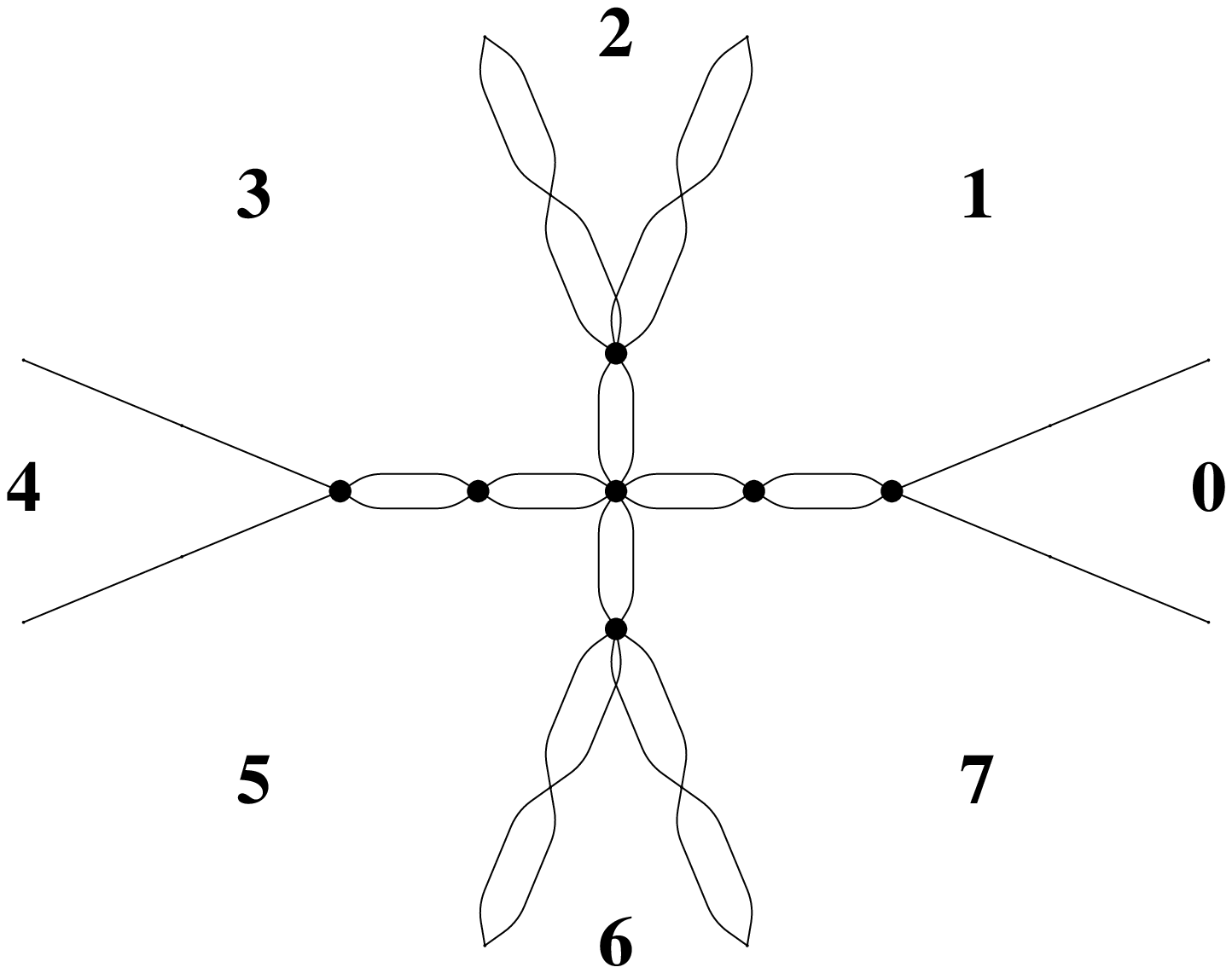}}
  \subfloat[(b)]{\includegraphics[width=0.49\textwidth]{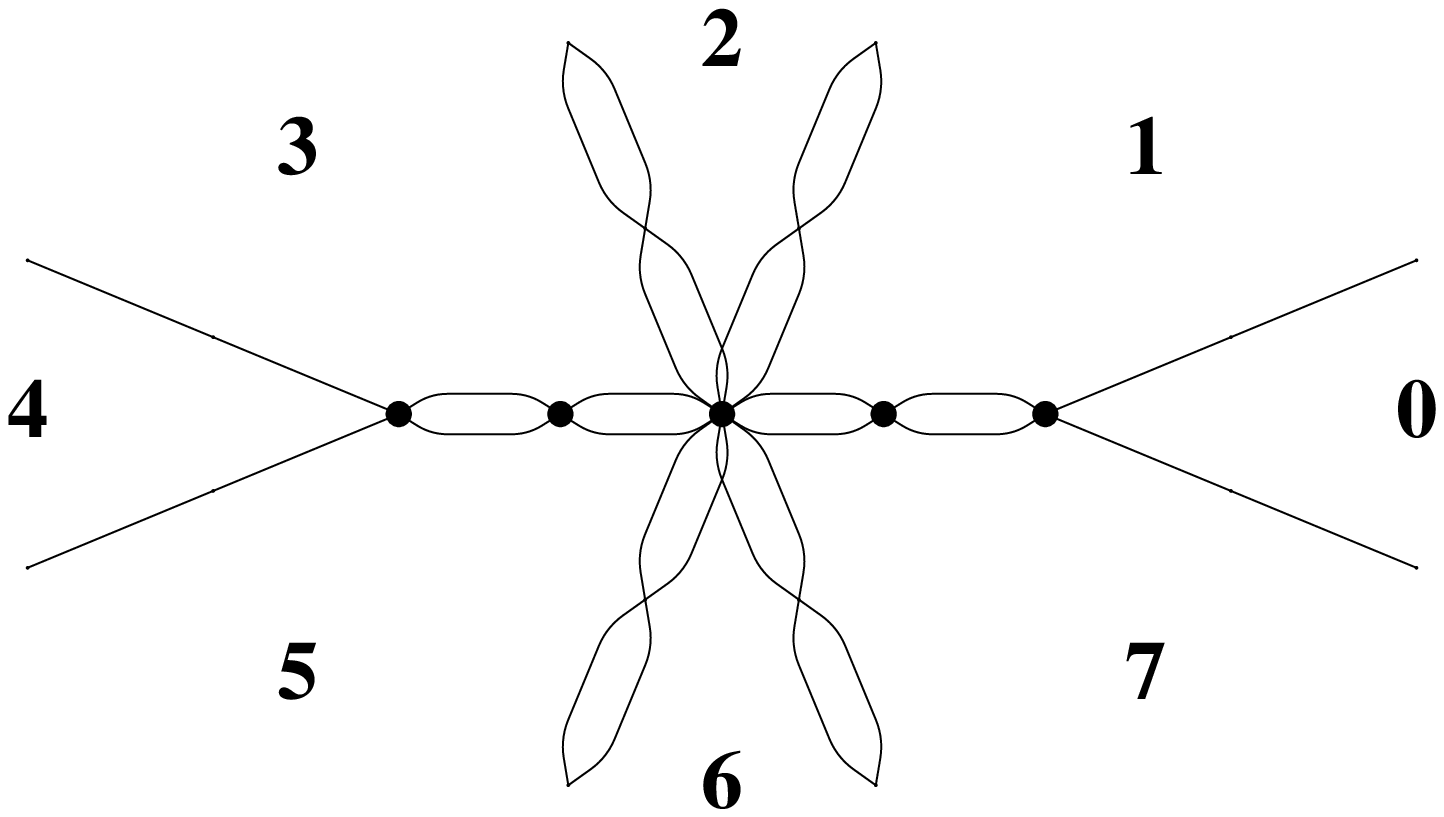}}
 \caption{The graphs $\Gamma$ and  $\actE_1^2(\Gamma)$}
\label{fig:action1}
\end{figure}
Applying $\actE_1^2,$ moves the $I\mhyp$structure at the $1\mhyp$junction one step to the right,
following the $1\mhyp$edges. Similarly, the $I\mhyp$structure at the $5\mhyp$junction moves one step to the left.
Therefore, $\actE_1^2(\Gamma)$ is given in Fig.~\ref{fig:action1}b. 
The graph $\actE_1^2(\Gamma)$ is now in ivy form, it consists of a center junction connected to four 
$I\mhyp$structures and two $Y\mhyp$structures. 
We proceed by using the algorithm in Lemma \ref{lemma:contracttwo}, and apply $\actE_1^2$ two times more.
These steps are given in Fig.~\ref{fig:action2}.
\begin{figure}[!ht]
\centering
  \subfloat[(a)]{\includegraphics[width=0.49\textwidth]{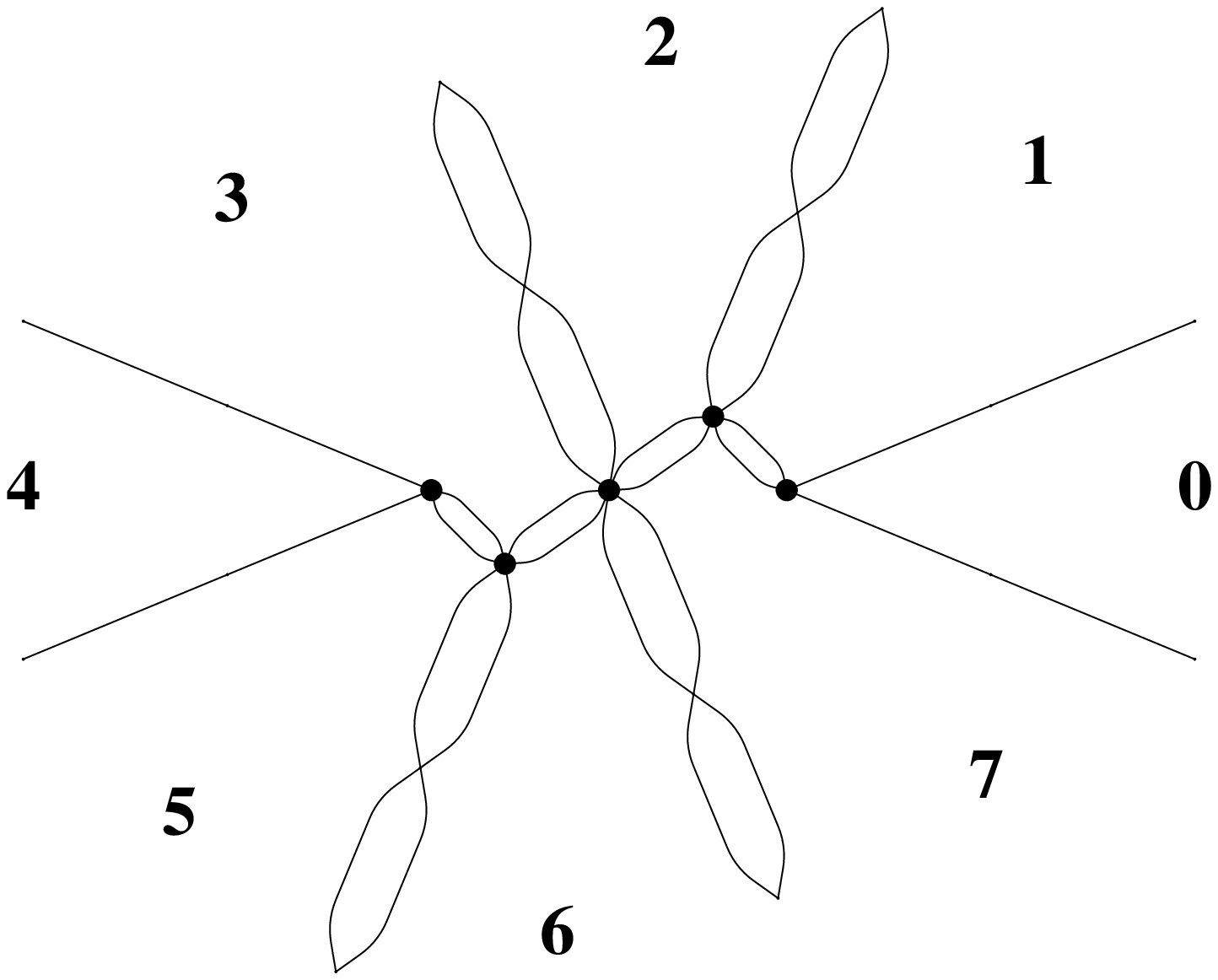}}
  \subfloat[(b)]{\includegraphics[width=0.49\textwidth]{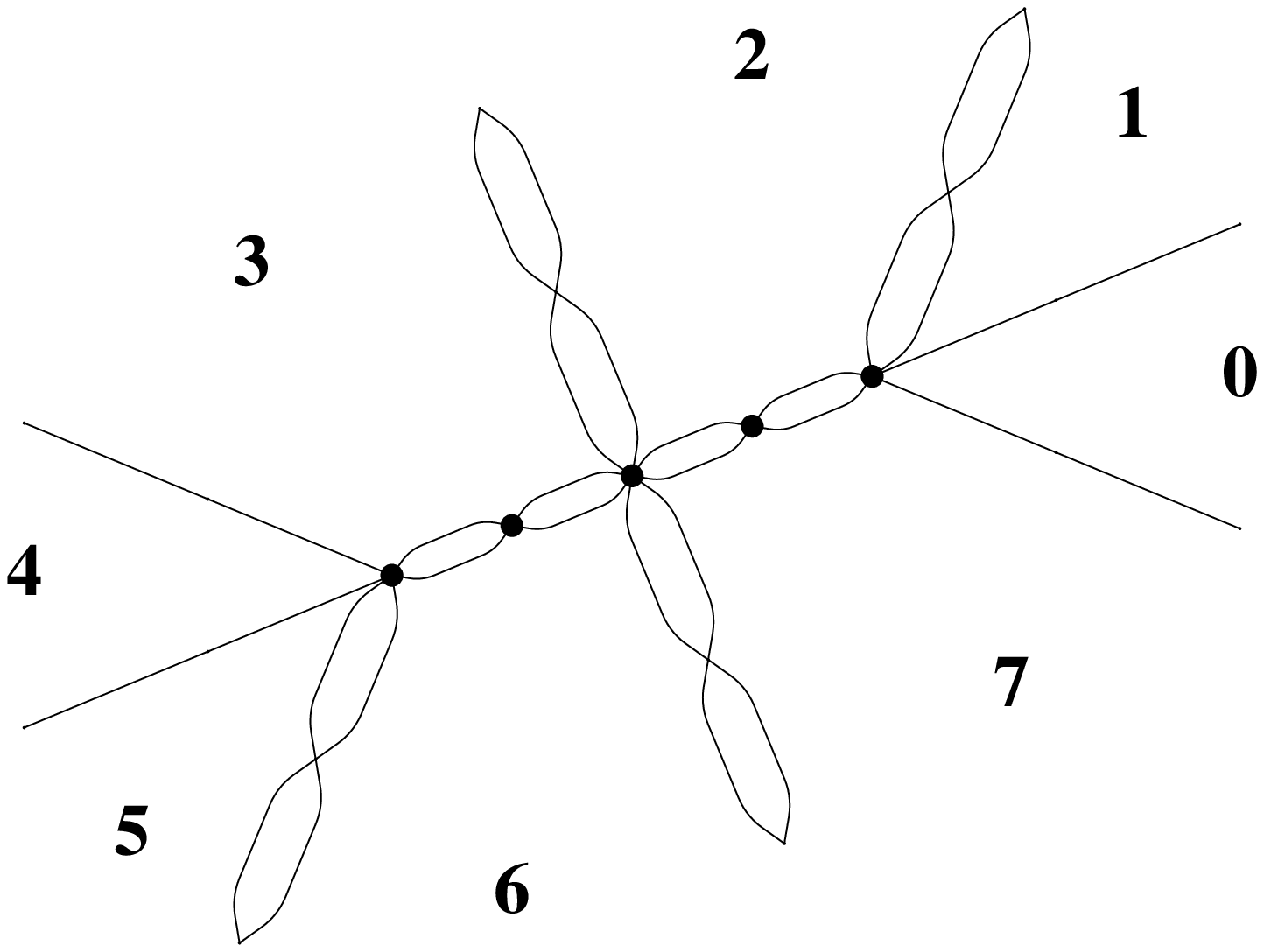}}
 \caption{The graphs $\actE_1^4(\Gamma)$ and  $\actE_1^6(\Gamma)$}
\label{fig:action2}
\end{figure}
The next step in the lemma is to move the newly created $V\mhyp$structures to the center junction.
We therefore apply $\actE_3^2$ two times. These final steps are presented in Fig.~\ref{fig:action3},
and we have reached the unique graph with only one junction.
\begin{figure}[!ht]
\centering
  \subfloat[(a)]{\includegraphics[width=0.49\textwidth]{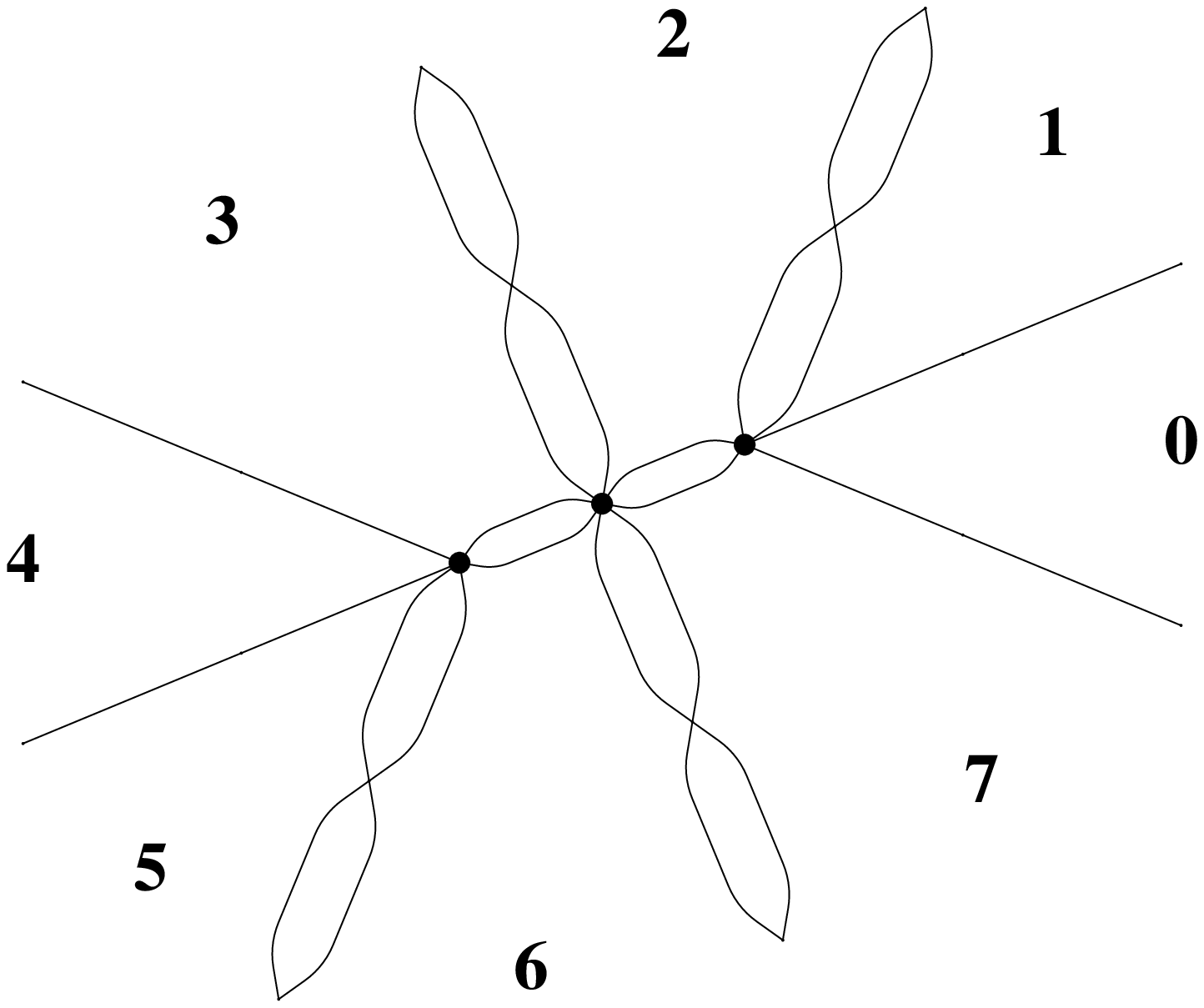}}
  \subfloat[(b)]{\includegraphics[width=0.49\textwidth]{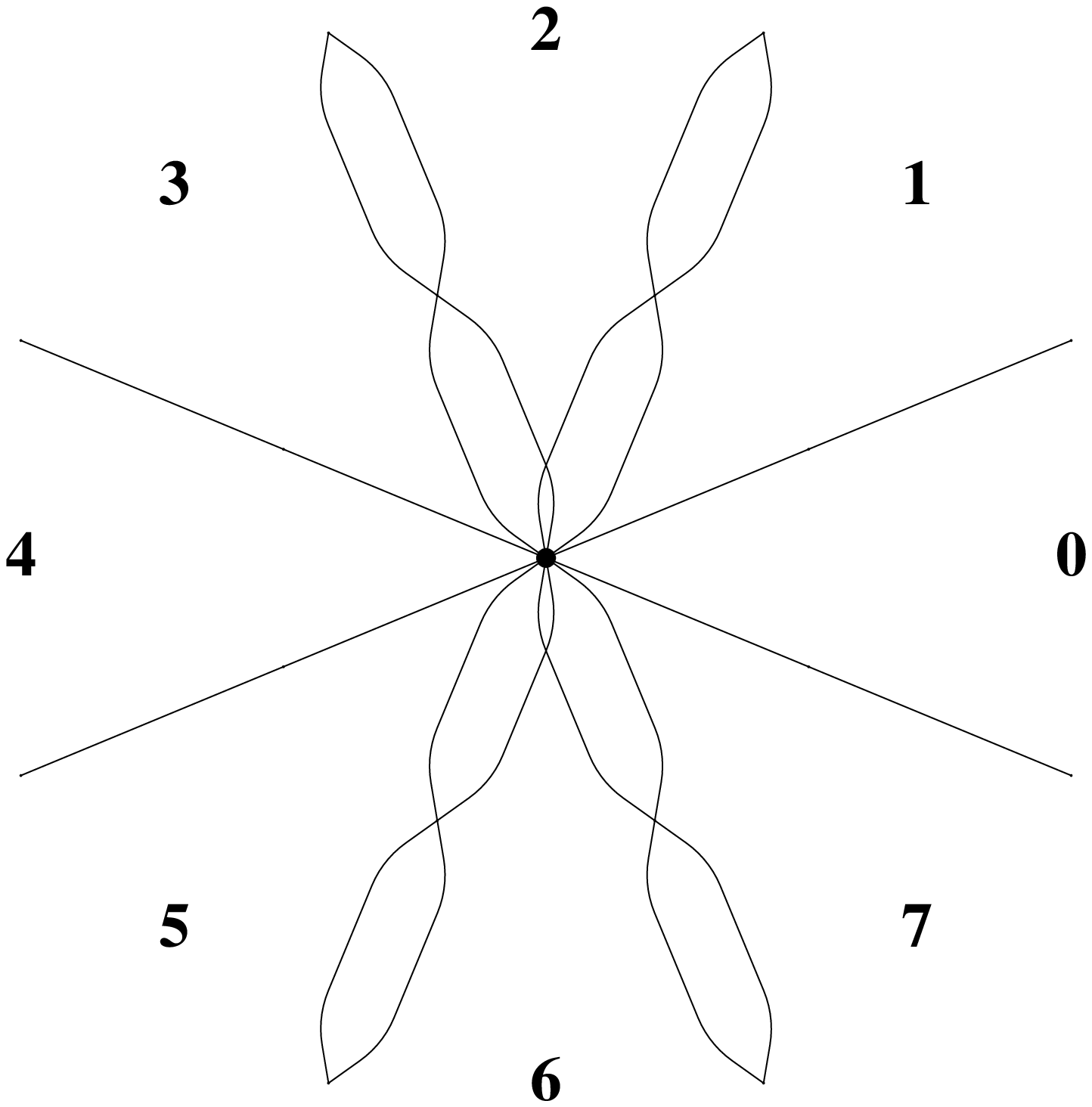}}
 \caption{The graphs $\actE_3^2\actE_1^6(\Gamma)$ and  $\actE_3^4\actE_1^6(\Gamma)$}
\label{fig:action3}
\end{figure}
